\documentclass[aps,a4paper,twocolumn]{revtex4}
\usepackage[utf8]{inputenc}
\usepackage{booktabs}
\usepackage{mathtools}
\usepackage{amsmath}
\usepackage{physics}
\usepackage[colorlinks,bookmarks=false,linkcolor=blue,urlcolor=blue,citecolor=magenta,linktocpage=true]{hyperref}
\usepackage[dvipsnames]{xcolor}
\usepackage{amssymb}
\usepackage{amsthm}
\usepackage[paperwidth=210mm,paperheight=297mm,centering,hmargin=2.0cm,vmargin=2.5cm]{geometry}
\usepackage{cancel}
\usepackage{bbm}
\usepackage{dsfont}

\DeclareMathAlphabet{\mathbbb}{U}{bbold}{m}{n}

\newtheorem{lemma}{Lemma}

\def \sing {\mathfrak{s}}
\def \compls {\mathfrak{c}}
\newcommand{\Mdu}[3]{{#1}_{#2}^{\phantom{#2}#3}}

\usepackage{soul}
\usepackage{empheq}

\begin{document}

\title{The Virasoro Completeness Relation and Inverse Shapovalov Form}
\author{Jean-François Fortin}
\email{jean-francois.fortin@phy.ulaval.ca}
\author{Lorenzo Quintavalle}
\email{lorenzo.quintavalle.1@ulaval.ca}
\author{Witold Skiba}
\email{witold.skiba@yale.edu}
\affiliation{$^\star$$^{\dagger}$ D\'epartement de Physique, de G\'enie Physique et d'Optique, Universit\'e Laval, Qu\'ebec, QC~G1V~0A6, Canada,\\
$^\ddagger$ Department of Physics, Yale University, New Haven, CT 06520, USA}

\begin{abstract}
In this work, we introduce an explicit expression for the inverse of the symmetric bilinear form of Virasoro Verma modules, the so-called Shapovalov form, in terms of singular vector operators and their conformal dimensions. Our proposed expression also determines the resolution of the identity for Verma modules of the Virasoro algebra, and can be thus employed in the computation of Virasoro conformal blocks via the sewing procedure.
\end{abstract}
\maketitle

\section{Introduction}

The central extension of the complexified Lie algebra of polynomial vector fields on the circle, the well-known Virasoro algebra, is at the heart of many aspects of modern theoretical physics. Since the seminal results of Kac~\cite{Kac:1978ge} and Feigin and Fuchs~\cite{Feigin:1981st,FeiginFuchs:book}, 
which clarified the conditions for Virasoro representations to be degenerate, 
it has been clear how a good control of the Virasoro representation theory~\cite{Iohara:book} can have a profound impact on our understanding of two-dimensional conformal field theories~\cite{DiFrancesco:1997nk}. The presence of degenerate fields in the spectrum is at the core of the celebrated exact solvability of the Virasoro minimal models~\cite{Belavin:1984vu}, Liouville theory~\cite{Dorn:1994xn,Zamolodchikov:1995aa,Teschner:1995yf}
and the recently-explored solvability of critical loop models~\cite{Nivesvivat:2023kfp,Roux:2024ubh}. These are important non-perturbative lampposts within the landscape of quantum field theory, that also shed light on our modern understanding of quantum gravity~\cite{Polchinski:books,Maldacena:1997re}. 

While the overall construction of degenerate Virasoro representations is well known, being equivalent to quotients of a Verma module by its proper submodules generated by singular vectors, an explicit understanding of certain related structures has not yet been achieved.
Exemplary in this sense is the case of the singular vectors themselves, for which a general formula beyond the first simplest cases~\cite{Benoit:1988aw,Millionshchikov2016} is still lacking. 

Another related example, central to this paper, is the so-called Shapovalov form, also commonly referred to as the Virasoro Gram matrix. This is the Hermitian form that naturally encodes the scalar product between states of a Verma module, whose inverse appears explicitly as the ``propagator'' of the sewing procedure to construct arbitrary conformal blocks~\cite{Sonoda:1988mf}. For this reason, the Shapovalov form can also constitute an important ingredient in the context of the AGT correspondence~\cite{Alday:2009aq} or the conformal bootstrap~\cite{Ribault:2014hia}. Despite its relevance, a general expression of the Shapovalov form is currently not known.  

In this paper, we aim to address at least in part this shortcoming by providing a way to construct, for generic central charge, the inverse Shapovalov in terms of singular vectors. This is achieved through the intimate relation between the inverse Shapovalov and the resolution of the identity (a.k.a. completeness relation) for the Virasoro Verma modules. The result that we get is in a form readily employable in the computation of conformal blocks, provided knowledge of singular vectors. 

The structure of this paper is as follows. In Section~\ref{sect:VirasoroAlgebra} we set up our notation for the Virasoro algebra and provide a basic definition of both the Virasoro singular vectors and the Shapovalov form. In Section~\ref{sect:inverseShapovalov_and_proof} we then discuss the connection between the resolution of the identity and the inverse Shapovalov form, to finally introduce and prove a new analytic expression of these objects in terms of singular vectors. 

\section{The Virasoro algebra}\label{sect:VirasoroAlgebra}

The Virasoro algebra, denoted as $\mathit{Vir}$, is an infinite-dimensional Lie algebra generated by $L_n$, $n\in \mathbb{Z}$ and a central element~$\hat{c}$ characterized by the following commutation relations:
\begin{gather}
    \comm{L_m}{L_n}=(m-n)L_{m+n}+\frac{\hat{c}}{12}(m^3-m)\delta_{m+n,0}\,, \notag\\ \comm{\hat{c}}{L_n}=0\,.
    \label{eq:Vir_comm_rel}
\end{gather}
A Verma module $V(h,c)$ of the Virasoro algebra $\mathit{Vir}$ corresponds to the module generated by a lowest-weight vector $\ket{h,c}$ with properties
\begin{equation}
\begin{aligned}
    L_0 \ket{h,c} &= h \ket{h,c} , \\ \hat{c} \ket{h,c} &= c \ket{h,c}, \\ L_n \ket{h,c} &= 0 \quad \text{for} \quad n>0\,.
\end{aligned}
\end{equation}
In physics language, these correspond to the set made of a primary state with generic conformal dimension and all of its descendants. 
For compactness of notation, we will often suppress the dependence of the states on the central charge $\ket{h}\equiv \ket{h,c}$ and trade the parameter $c$ with the more convenient parametrization of the central charge in terms of $t$ defined by
\begin{equation}
    c=13-6 t-\frac{6}{t}.
\end{equation}
The Verma module $V(h,c)$ is therefore a space spanned by $\ket{h}$ and the action on it of the negative modes of $\mathit{Vir}$. Generic vectors can be then expressed as linear combinations of monomials of the form
\begin{equation}
    L_{-i_m}L_{-i_{m-1}}\dots L_{-i_1} \ket{h} \quad \text{with} \quad i_1,i_2,\dots ,i_m > 0.
    \label{eq:ordered_monomials}
\end{equation}
For every monomial above there is a well-defined grading provided by $\ell := \sum_{j=1}^m i_{j}$ which is known as the \emph{level} of the monomial. For fixed level $\ell$, the set of all monomials of form~\eqref{eq:ordered_monomials} is isomorphic to the set of ordered integer partitions of $\ell$
\begin{equation}
    \mathbb{O}_\ell=\left \{ (i_1,\dots, i_m) \in \mathbb{Z}_{+}^m \,\bigg\vert\, 0<m\le \ell \And\! \sum_{j=1}^m i_j = \ell \right\},
    \label{eq:set_ordered_partitions}
\end{equation}
but due to the commutation relations~\eqref{eq:Vir_comm_rel}, this set is overcomplete.

A basis is instead provided by the set of all ordered monomials
\begin{multline}
    L_{-\mu_m}L_{-\mu_{m-1}}\dots L_{-\mu_1}\! \ket{h}\\
    \quad \text{with} \quad \mu_m\!\ge\! \mu_{m-1}\! \ge \!\dots \ge\! \mu_1 > 0,
    \label{eq:standard_monomial_basis}
\end{multline} 
which, for a given level $\ell$, is isomorphic to the set of standard integer partitions of $\ell$ (equivalently, to the space of Young diagrams with $\ell$ boxes)
\begin{multline}
    \mathbb{P}_\ell = \biggl \{ (\mu_1,\dots, \mu_m) \in \mathbb{Z}_{+}^m \,\bigg\vert\, 0<m\le \ell  \\
    \And \mu_m\ge \mu_{m-1}\ge \dots \ge \mu_1 >0 \And \sum_{j=1}^m \mu_j = \ell \biggr\}.
    \label{eq:partitions_P}
\end{multline}
We will use Greek letters for labels that take values in $\mathbb{P}_\ell$, which will allow us to display basis vectors of the Verma module $V(h,c)$ as $L_{-\mu} \ket{h}$. Here the minus sign is understood to reverse both the order and all of the signs of the vectors $(\mu_1,\dots,\mu_m)$ in $\mathbb{P}_\ell$. For a given monomial labeled by $\mu$, we denote its level as
\begin{equation}
    \abs{\mu}=\sum_{j=1}^m \mu_j \qquad \text{with} \quad \mu = (\mu_1,\dots,\mu_m)\,.
\end{equation}

Not all of the irreducible representations of $\mathit{Vir}$ correspond to Verma modules. In the following section we discuss the other irreducibles via the introduction of the pairing provided by the Shapovalov form. 

\subsection{Singular Vectors and the Shapovalov Form}

When working with the Virasoro algebra, it is natural to introduce a Hermitian conjugation that maps positive modes to negative ones as $L_n^\dagger= L_{-n}$. Via this conjugation, the ordered products of generators $L_{-\mu}$ that characterize the basis vectors of a Verma module are mapped to the products $L_{\mu}$, with an ordering that matches that of~\eqref{eq:partitions_P}. With this at hand, one can use the commutations relations of Virasoro to rewrite any product of positive modes with negative modes as follows
\begin{equation}
    L_\mu L_{-\nu}=L_{-\alpha}\Mdu{S}{\mu\nu}{\alpha\beta}(L_0,\hat{c})L_\beta \, .
\end{equation}
For the case $\alpha=\beta=()$, only possible when $\abs{\mu}=\abs{\nu}=\ell$, this corresponds to an important operator known as the \emph{Shapovalov form} $\left[S_{\ell}(L_0,\hat{c})\right]_{\mu\nu}$.
The Shapovalov form can be explicitly computed via the commutation relations of Virasoro generators and focusing on terms that only include the zero-level generators $L_0$ and $\hat c$. For instance at level 3 the $\mu=(3)$, $\nu=(1,2)$ term is
\begin{multline}
    \left[L_{3}L_{-2}L_{-1}\right]\big\rvert_{\text{zero}} = \!\left[(5L_{1}+\cancel{L_{-2}L_{3}})L_{-1}\right] \big\rvert_{\text{zero}}\\
    =\left[10L_0 + \cancel{5L_{-1}L_{1}}\right]\big\rvert_{\text{zero}}=10 L_0\,,
    \end{multline}
and the overall matrix can be similarly computed to be
\begin{equation*}
\begin{aligned}
    &S_3(L_0,\hat{c})=\\
&\scalebox{0.89}{$
    \left(\!
\begin{array}{ccc}
\! 6 L_0-12 \frac{\hat t^2+1}{\hat t}+26 & 10 L_0 & 24 L_0 \\
 10 L_0 & \hspace*{-11pt} L_0 \left(8 L_0-6 \frac{\hat t^2+1}{\hat t}+21\right) & 12 L_0 (3 L_0+1) \\
 24 L_0 & 12 L_0 (3 L_0+1) & \hspace*{-6pt}24 L_0 (L_0+1) (2 L_0+1) \\
\end{array}
\!\right)\!.$
}\!
\end{aligned}\!
\end{equation*}
The importance of the Shapovalov form stems from the fact that it encodes a scalar product between states of a Verma module. If we consider in fact a scalar product associated with the Hermitian conjugation discussed before and such that $\braket{h}=1$, we have that the scalar product between descendant states is given by
\begin{equation}
    \matrixel{h}{L_{\mu}L_{-\nu}}{h}=\matrixel{h}{\left[S_{\ell}(L_0,\hat c)\right]_{\mu\nu}}{h}=\left[S_{\ell}(h,c)\right]_{\mu\nu}
\end{equation}
if $\abs{\mu}=\abs{\nu}=\ell$, while the matrix elements vanish identically for vectors of different levels
\begin{equation}
    \matrixel{h}{L_{\mu}L_{-\nu}}{h}=0 \qquad \text{if} \quad \abs{\mu}\ne \abs{\nu}\,.
\end{equation}
As opposed to $S_\ell(L_0,\hat c)$, the Shapovalov matrix $S_\ell (h,c)$ need not be of full rank. For instance, in the case $h=0$ one can easily check that $S_1(0,c)=\matrixel{0}{L_1L_{-1}}{0}=\matrixel{0}{2L_0}{0}$=0, which implies that the vector $L_{-1}\ket{0}$ is orthogonal to any other state in the Verma module. Vectors of this kind, which can be shown to be combinations of descendants $v^{\mu} L_{-\mu} \ket{h}$ that also satisfy the primary condition $L_n (v^{\mu} L_{-\mu} \ket{h})=0$ for all $n>0$, take the name of \emph{singular vectors} (the term \emph{null vector} is also commonly used). Verma modules that contain singular vectors are also known as \emph{degenerate Verma modules}. The presence of singular vectors is detected by the Shapovalov matrix which, above a certain level, becomes degenerate for values of $h$ that admit singular vectors. 
To study these, one can work out the determinant of the Shapovalov matrix, which for any level $\ell$ is described by the \emph{Kac determinant formula}
\begin{multline}
    \det S_\ell(h,c)= \prod_{\substack{r,s \ge 1\\r s \le \ell}} \left[(2r)^s s!\right]^{p(\ell-r s)-p(\ell-r(s+1))}\\
    \times \prod_{\substack{r,s \ge 1\\r s \le \ell}} \left[h-h_{\expval{r,s}} \right]^{p(\ell-r s)},
    \label{eq:Kacdeterminant}
\end{multline}
where
\begin{equation}
    h_{\expval{r,s}}= \frac{t  (r+1)-(s+1)}{2} \frac{t^{-1}(1-s)-(1-r)}{2}.
\end{equation}
The Kac determinant instructs us that, for generic central charge, there are new possible singular vectors at any level~$\ell$. These are parametrized by two integers $r,s$ such that $r s = \ell$, and each vector occurs when the primary state has conformal dimension equal to the corresponding $h_{\expval{r,s}}$. Explicitly, this set is in one-to-one correspondence with
\begin{equation}
    \sing_\ell = \left\{ (r,s) \in \mathbb{Z}_{+}^2 \,\big \vert\, r s = \ell \right\},
    \label{eq:set_singular_vectors}
\end{equation}
which we will then refer to as the set of singular vectors at level $\ell$.
The remaining zeros of the Kac determinant, \textit{i.e.}~the $h_{\expval{r,s}}$ for $r s < \ell$, correspond instead to descendants of singular vectors encountered at level $r s <\ell$. 

For every $(r,s)\in \mathfrak{s}_\ell$, there is then a singular vector of the form
\begin{equation}
    L_{\expval{r,s}}\! \ket{h_{\expval{r,s}}}\equiv v_{\expval{r,s}}^\mu L_{-\mu}\! \ket{h_{\expval{r,s}}}\,,
    \label{eq:singular_vector_decomposition}
\end{equation}
where $v_{\expval{r,s}}^\mu$ are coefficients that depend on $r$, $s$, $t$ which multiply the basis of products of generators at level $|\mu|=r s$. In an abuse of notation, we will often use the term singular vector also to stand for the combination of generators $L_{\expval{r,s}}$ alone, without acting on a state. Conventionally, we normalize the singular vectors such that the coefficient of the term containing $L_{-1}^{r s}$ is equal to one, so
\begin{equation}
    v_{\expval{r,s}}^{(1,\dots,1)}=1\,.
\end{equation}

While concrete expressions for singular vectors for $r=1,2$ or $s=1,2$ are known explicitly~\cite{Benoit:1988aw,Millionshchikov2016}, only some properties~\cite{FeiginFuchs:book} and computation algorithms~\cite{Bauer:1991ai,Kent:1991qj} are known for the general case. The main formulas we will introduce in Section~\ref{ssect:inverse_Shapovalov} will be expressed in terms of the $L_{\expval{r,s}}$, and they thus rely on the knowledge of singular vectors. 

To show one concrete example, let us consider the singular vector for $\expval{r,s}=\expval{3,1}$, corresponding to
\begin{equation*}
\begin{aligned}
    L_{\expval{3,1}}&\ket{h_{\expval{3,1}}}\\
    &= \left[2 t (2 t-1) L_{-3}-4 t L_{-2} L_{-1}+L_{-1}^3\right]\ket{h_{\expval{3,1}}}\,.
    \end{aligned}
\end{equation*}
This behaves as a lowest-weight vector since, when acted upon with positive modes and the product of generators is appropriately reordered, the combination of zero modes that are produced vanish when acting on $\ket{h_{\expval{3,1}}}$. For instance, acting on it with $L_1$, we get
\begin{equation}
    \begin{aligned}
    L_1 L_{\expval{3,1}}& \ket{h_{\expval{3,1}}}\\
    &= -2\left(4t L_{-2}-3 L_{-1}^2\right) \left(L_0-2t+1\right) \ket{h_{\expval{3,1}}}
    \end{aligned}
    \label{eq:vanishing_of_L31}
\end{equation}
which clearly vanishes due to $h_{\expval{3,1}}=2t-1$.

Note how, in the last operation, we needed to perform the reordering in order to see the vanishing action on the singular vector.
This is going to be relevant in the remainder of this paper, since the resolution of the identity will require us to work in a natural extension of the Universal Enveloping Algebra that includes rational functions of the zero modes of the algebra.
In this case, apparent ambiguous expressions may appear, such as
\begin{equation}
    L_1 L_{\expval{3,1}} \frac{1}{L_0-h_{\expval{3,1}}} \ket{h_{\expval{3,1}}}.
    \label{eq:apparent_divergence}
\end{equation}
While seemingly divergent, expressions like the one above can be made sense of by implementing the simple general prescription to \emph{always reorder and recombine the combinations of generators before acting on a state}. In the case of~\eqref{eq:apparent_divergence}, for example, our prescription requires us to perform the same steps of~\eqref{eq:vanishing_of_L31} to get
\begin{equation}
    \begin{aligned}
        L_1 &L_{\expval{3,1}} \frac{1}{L_0-h_{\expval{3,1}}} \ket{h_{\expval{3,1}}}\\
        &= -2\left(4t L_{-2}-3 L_{-1}^2\right) \frac{L_0-2t+1}{L_0-h_{\expval{3,1}}}\ket{h_{\expval{3,1}}}\\
        &= -2\left(4t L_{-2}-3 L_{-1}^2\right) \ket{h_{\expval{3,1}}}
    \end{aligned}
\end{equation}
which is well defined.

\section{The inverse Shapovalov form and the resolution of the identity}\label{sect:inverseShapovalov_and_proof}

We now turn to the central subject of this work: the inverse Shapovalov form and its relation with the resolution of the identity.
As the name suggests, the inverse Shapovalov form $S_\ell^{-1}(L_0,\hat c)$ is simply the matrix inverse of the object we defined in the previous section, characterized by the obvious identities
\begin{multline}
    \left[S_\ell^{-1}(L_0,\hat c) \right]^{\mu\nu}\left[S_\ell(L_0,\hat c) \right]_{\nu\rho}\\
    =\left[S_\ell(L_0,\hat c) \right]_{\rho\nu}\left[S_\ell^{-1}(L_0,\hat c) \right]^{\nu\mu}=\delta^{\mu}_\rho\,.
\end{multline}
Given the positioning of indices in the inverse Shapovalov form, it is natural to work with it in an index-free fashion, by simply contracting all indices with basis vectors as follows
\begin{equation}
    \mathbf{S}_\ell^{-1}(L_0,\hat c) = L_{-\mu}\left[S_\ell^{-1}(L_0,\hat c) \right]^{\mu\nu}L_{\nu}\,.
    \label{eq:isf}
\end{equation}
It will soon be clear how this is a very natural object to consider, given the simple expression we propose for $\mathbf{S}_\ell^{-1}(L_0,\hat c)$, and the way it appears naturally in the resolution of the identity within a Virasoro Verma module. 
Note that, given an index-free expression for $\mathbf{S}_\ell^{-1}(L_0,\hat c)$, one can obtain any matrix element $\left[S_\ell^{-1}(L_0,\hat c) \right]^{\mu\nu}$ by reordering the products of generators and extracting the coefficient of the $L_{-\mu}L_{\nu}$ term.

The relevance of all this to the resolution of the identity begins to manifest with the observation that the index-free inverse Shapovalov $\mathbf{S}_{\ell}^{-1}(L_0,\hat{c})$ is an operator that acts as the identity on any level-$\ell$ descendant. For $\abs{\rho}=\ell$, we have in fact, given any primary state $\ket{\psi}$ of a Verma module $V(h,c)$, the following identity:
\begin{equation}
\begin{aligned}
    \mathbf{S}_{\ell}^{-1}(L_0,\hat{c})L_{-\rho} \ket{\psi} = L_{-\mu}\left[S^{-1}_{\ell}(L_0,\hat{c})\right]^{\mu\nu} \!L_{\nu} L_{-\rho} \ket{\psi} \\
    = L_{-\mu}\left[S^{-1}_{\ell}(L_0,\hat c)\right]^{\mu\nu}\left[S_{\ell}(L_0,\hat c)\right]_{\nu\rho}\ket{\psi} = L_{-\rho}\! \ket{\psi},
    \label{eq:inverseShap_identity_on_descendants}
\end{aligned}
\end{equation}
where we used the defining properties of the Shapovalov form and its inverse. Note that expression~\eqref{eq:inverseShap_identity_on_descendants} is valid for any Verma module, including degenerate Verma modules $V(h_{\expval{r,s}},c)$ for which $S_{\ell}(h_{\expval{r,s}},c)$ is not invertible. Our prescription of always having the products of generators in a standard ordering before acting on a state, in fact, prevents us from acting on $\ket{h_{\expval{r,s}}}$ with $S_{\ell}^{-1}(L_0,\hat c)$, which could otherwise produce divergences. 

With this in mind, we can then write the resolution of the identity $\mathds{1}(h,c)$ for a generic Verma module $V(h,c)$ as
\begin{equation}
    \mathds{1}\equiv\mathds{1}(h,c)\, = \sum_{\ell\geq0} L_{-\mu}\ket{h}\!\bra{h}\!\left[S^{-1}_{\ell}(L_0,\hat{c})\right]^{\mu\nu} L_{\nu}\,,
    \label{eq:resol_identity}
\end{equation}
where it is understood that $L_{()}=S_0^{-1}(L_0,\hat{c})=1$ and we introduced the projector $\ket{h}\!\bra{h}$
to kill the contribution of the level-$\ell$ inverse Shapovalov when acting on terms of level $n>\ell$. 

Acting with the representation~\eqref{eq:resol_identity} on an arbitrary state $L_{-\rho} \ket{h}$, we can indeed verify
\begin{equation}
\begin{aligned}
    &\sum_{\ell\geq0} L_{-\mu}\ket{h}\!\bra{h}\left[S^{-1}_{\ell}(L_0,\hat c)\right]^{\mu\nu} L_{\nu} L_{-\rho} \ket{h}\\
    &=\sum_{\ell\geq0} L_{-\mu}\ket{h}\!\bra{h}\left[S^{-1}_{\ell}(L_0,\hat c)\right]^{\mu\nu}\left[S_{\ell}(L_0,\hat c)\right]_{\nu\rho}\ket{h} \delta_{\ell,\abs{\rho}} \\
    &= L_{-\rho}\! \ket{h},
\end{aligned}
\label{eq:applying_general_resolution_identity}
\end{equation}
where we used the fact that $\matrixel{h}{L_\nu L_{-\rho}}{h}$ is zero unless $\ell\equiv\abs{\nu}=\abs{\rho}$.
Let us note how, given~\eqref{eq:inverseShap_identity_on_descendants}, the expression~\eqref{eq:resol_identity} can also be adapted to work in the case of degenerate Verma modules with the added precaution of considering a projector that can only act on its right (as opposed to being acted upon with the inverse Shapovalov form).

It is now clear that the index-free version of the inverse Shapovalov form is directly connected to the expression of the resolution of the identity for Virasoro Verma modules. 
Note that, in the context of two-dimensional conformal field theories, this is relevant to the computation of arbitrary correlation functions of fields $\phi_i(z_i)$ through the sewing relation~\cite{Sonoda:1988mf}
\begin{multline}
    \expval{\phi_1(z_1)\cdots\phi_n(z_n)}=\sum_h\matrixel{0}{\phi_1(z_1)\cdots\phi_i(z_i)}{L_{-\mu}h}\\
    \times\left[S^{-1}_{\ell}(h,c)\right]^{\mu\nu}\matrixel{L_{-\nu}h}{\phi_{i+1}(z_{i+1})\cdots\phi_n(z_n)}{0}.
    \label{eq:correlator_decomp_shapovalov}
\end{multline}
Equation~\eqref{eq:correlator_decomp_shapovalov} effectively translates into a conformal block decomposition, where the higher-point conformal blocks are determined in terms of lower-point ones and the inverse Shapovalov matrix.

Our next goal is to provide an expression for the (index-free) inverse Shapovalov form~\eqref{eq:isf} in terms of the singular vectors, which will consequently make also the expression~\eqref{eq:resol_identity} explicit.

\subsection{An Expression for the Inverse\\ Shapovalov Form}\label{ssect:inverse_Shapovalov}

Our main goal now is to provide an explicit expression for the inverse Shapovalov form $\mathbf{S}_\ell^{-1}(L_0,\hat c)$. Our expression will be written and proven for generic values of the central charge, i.e. $t\in \mathbb{C}\setminus\mathbb{Q}$, but a limiting procedure to rational values of $t$ is in principle possible.
To write down our proposal, we need to first discuss notation. 
In addition to the set $\mathfrak{s}_\ell$ defined in~\eqref{eq:set_singular_vectors} that is used to label the singular vectors $L_{\expval{r,s}}$ at level $\ell$, we also need a set $\compls_{\ell}$ that, for every integer number $\ell\ge 0$, labels all possible products of singular vectors that are of total level equal to $\ell$. Given the set of order-discriminating integer partitions $\mathbb{O}_\ell$ defined in~\eqref{eq:set_ordered_partitions}, the set $\compls_\ell$ can be written as 
\begin{multline}
    \compls_\ell= \Bigl \{ \bigl((r_1,s_1),\dots, (r_m,s_m)\bigr) \in \sing_{i_1}\times \dots \times \sing_{i_m} \\\Big\vert\, 0<m\leq \ell\And (i_1,\dots, i_m) \in \mathbb{O}_\ell \Bigr \}.
\end{multline}
Note that, compatibly with our definition, we have the empty set at level zero, $\compls_0=\{\}$.
To compactify our notation, we will use a different slicing of the set $\compls_\ell$ made of all vectors $\boldsymbol{r}\equiv(r_1,\dots,r_m)$ and $\boldsymbol{s}\equiv(s_1,\dots,s_m)$ subject to the constraint $\boldsymbol{r}\cdot\boldsymbol{s}=\ell$; we will then just display the latter constraint when summing over this set $\compls_{\ell}$.

At this point, we can write our proposed expression for the inverse Shapovalov form as
\begin{multline}
    \mathbf{S}_\ell^{-1}(L_0,\hat{c})=\sum_{\boldsymbol{r}\cdot\boldsymbol{s}=\ell} L_{\expval{r_{m},s_{m}}}\dots L_{\expval{r_{1},s_{1}}}\\
    \times\frac{q_{\expval{r_1,s_1},\dots, \expval{r_{m},s_{m}}}}{L_0-h_{\expval{r_1,s_1}}}L_{\expval{r_{1},s_{1}}}^{\dagger}\dots L_{\expval{r_{m},s_{m}}}^{\dagger} \,,
    \label{eq:inverse_Shapovalov_form}
\end{multline}
in terms of the coefficients $q_{\expval{r_1,s_1},\dots \expval{r_{m},s_{m}}}$ defined by
\begin{multline}
    q_{\expval{r_1,s_1},\dots, \expval{r_{m},s_{m}}}\!=\\
    \! 
    \frac{\prod_{i=1}^{m} q_{\expval{r_i,s_i}}}{\prod_{j=2}^{m}\!\left(h_{\expval{r_{j-1},s_{j-1}}}\!+\!r_{j-1} s_{j-1}\!-\!h_{\expval{r_j,s_j\!}}\!\right)}\,,
    \label{eq:recurrence_q_coefficients}
\end{multline}
with
\begin{equation}
    q_{\expval{r,s}}=\lim_{h\to h_{\expval{r,s}}}\frac{h-h_{\expval{r,s}}}{\sum_{\substack{\mu,\nu\\|\mu|=|\nu|=rs}}v_{\expval{r,s}}^\mu[S_{rs}(h,c)]_{\mu\nu}v_{\expval{r,s}}^\nu}\,.
    \label{eq:qrs_as_inverse_norm}
\end{equation}
A few comments are in order. To begin with, the expression~\eqref{eq:inverse_Shapovalov_form} is reminiscent of a spectral decomposition for the inverse Shapovalov form but with an overcomplete basis. For instance, at level $\ell=2$, the overcomplete basis is given by $\{L_{\expval{2,1}},L_{\expval{1,2}},L_{\expval{1,1}}L_{\expval{1,1}}\}$, contrary to the standard basis $\{L_{-2},L_{-1}L_{-1}\}$.
Furthermore, the coefficients~\eqref{eq:qrs_as_inverse_norm} correspond to the inverse of the regularized norm squared for the singular vectors $L_{\expval{r,s}}$, interpreting the Shapovalov matrix $S_{rs}(h,c)$ as a metric in the space of descendants. An expression for this norm has been proposed by Zamolodchikov in~\cite{Zamolodchikov:2003yb}. This expression, which we display in~\eqref{eq:Zamolodchikov_norm}, is equivalent to the inverse of the following:
\begin{equation}
    q_{\expval{r,s}}\!=\!\left[\! \frac{2(-1)^{r s +1}}{r s} \!\prod_{j=1}^{r} (j t)_s(-jt)_s \!\prod_{k=1}^{s}\! (k/t)_r(-k/t)_r\!\right]^{-1} \hspace*{-5pt},
    \label{eq:single_vector_coefficients}
\end{equation}
which we find more pleasant in terms of the ranges of products and due to the absence of square roots of $t$. As a final remark, both~\eqref{eq:recurrence_q_coefficients} and~\eqref{eq:qrs_as_inverse_norm} display some singular behavior when $t=t_j^\pm$ (if the latter is different than zero or infinity), where
\begin{equation}
    t_j^+=-\frac{s_{j-1}+s_j}{r_{j-1}-r_j},\qquad t_j^-=-\frac{s_{j-1}-s_j}{r_{j-1}+r_j}.
    \label{eq:central_charge_divergences}
\end{equation}
These points correspond to rational values of $t$, for which the representations are known to be doubly-degenerate, where identities can arise between (products of) singular vector operators and the Shapovalov form develops double zeros.
These apparent singular points, however,
do not translate to a divergence of the index-free inverse Shapovalov form. In fact, to respect~\eqref{eq:inverseShap_identity_on_descendants}, the latter has to be non-divergent for arbitrary Verma modules, including those whose central charge satisfies~\eqref{eq:central_charge_divergences} for some pair of singular vectors. 
Heuristically, the singularities in the $q$ coefficients arise to allow, in the limit $t\to t_{j}^{\pm}$, the emergence of double poles $(L_0-h_{\expval{r,s}})^{-2}$ when certain (products of) singular vector operators are identified and a common denominator is taken.
While we did not come up with a general way to show this phenomenon explicitly, we could indeed verify for the first low levels that the apparent poles in~\eqref{eq:recurrence_q_coefficients} and~\eqref{eq:qrs_as_inverse_norm} get canceled when recombining all coefficients of the inverse Shapovalov in the standard monomial basis~\eqref{eq:standard_monomial_basis}. This can be seen as evidence that our formula can also be used to resolve the level-$\ell$ identity in cases with rational $t$, provided that the limit to the corresponding value is taken carefully.

With our expression for the inverse Shapovalov form, the resolution of the identity~\eqref{eq:resol_identity} can be explicited as
\begin{multline}
    \mathds{1}=
    \sum_{\substack{
    \ell\ge 0\\
    \boldsymbol{r}\cdot\boldsymbol{s}=\ell}} L_{\expval{r_{m},s_{m}}}\dots L_{\expval{r_{1},s_{1}}}\ket{h}\\
    \times\bra{h}\frac{q_{\expval{r_1,s_1},\dots, \expval{r_{m},s_{m}}}}{L_0-h_{\expval{r_1,s_1}}}L_{\expval{r_1,s_1}}^{\dagger}\dots L_{\expval{r_{m},s_{m}}}^{\dagger}.
    \label{eq:resol_identity_explicit}
\end{multline}

In the following subsection, we will prove the expressions~\eqref{eq:inverse_Shapovalov_form}--\eqref{eq:qrs_as_inverse_norm}, while we will discuss how to retrieve~\eqref{eq:single_vector_coefficients} from Zamolodchikov's formula~\eqref{eq:Zamolodchikov_norm} in Appendix~\ref{app:Zamolodchikov_proof}.

\subsection{Proving the Completeness Relation and the Inverse Shapovalov Form}

To begin with our proof of the expression for the inverse Shapovalov form, let us first recall two important results. The first one, proven by Brown in~\cite{BROWN2003160}, informs us of the pole structure of the inverse Shapovalov form.
\begin{lemma}
    \label{lemmaBrown}
    \textbf{\textup{(Brown)}} The inverse Shapovalov form $S_\ell^{-1}(h,c)$ has only simple poles in $h$ for generic values of the central charge.
\end{lemma}
\begin{proof}
    See Corollary 6.2 of~\cite{BROWN2003160}.
\end{proof}
Notice that the pole structure of \eqref{eq:inverse_Shapovalov_form} is in agreement with Lemma \ref{lemmaBrown}.  The second result concerns, instead, the way positive modes can annihilate descendants of singular vectors.
\begin{lemma}
    \label{lemma2}
    Given labels $\nu$ and $\rho$ with $\abs{\nu}>\abs{\rho}\ge 0$, we have
    \begin{multline}
        L_\nu L_{-\rho} L_{\expval{r,s}}\\[2pt]
        =b^\lambda_{\nu\rho}(h_{\expval{r,s}},c) L_{-\lambda}(L_0-h_{\expval{r,s}})^n +(\dots),
        \label{eq:Lemma2}
    \end{multline}
    where $b^\lambda_{\nu\rho}(h_{\expval{r,s}},c)$ are undetermined coefficients, the suppressed terms $(\dots)$ vanish when applied on any primary state, the repeated index satisfies $\abs{\lambda}=r s+\abs{\rho}-\abs{\nu}$, and the exponent $n$ is a positive integer.
\end{lemma}

\begin{proof}
The standard reordering of the product $L_{\nu}L_{-\rho}$ leads to a sum of terms $L_{-\alpha}L_\beta$ with $|\alpha|\geq0$ and $|\beta|=|\nu|-|\rho|+|\alpha|>0$, which has to vanish when acting on a singular vector $L_{\expval{r,s}}\ket{h_{\expval{r,s}}}$. This implies the presence of a $(L_0-h_{\expval{r,s}})^n$ factor with a positive power $n$, as happened for example in~\eqref{eq:vanishing_of_L31}. This power can eventually be fixed to one by the use of Lemma~\ref{lemmaBrown}, as seen below.
\end{proof}

To prove~\eqref{eq:inverse_Shapovalov_form}--\eqref{eq:qrs_as_inverse_norm} we will make extensive use of the identity action of $\mathbf{S}_{\ell}^{-1}(L_0,\hat{c})$ expressed in~\eqref{eq:inverseShap_identity_on_descendants}.
We will use strong induction in $\ell$ to prove that, when acting on any level-$\ell$ descendant, the RHS of~\eqref{eq:inverse_Shapovalov_form} acts as the identity.

The starting point of induction is simple. Ignoring $S_0(L_0,\hat{c})=1$ which trivially acts as the identity, we have at level $\ell=1$
\begin{multline}
    \mathbf{S}_{1}^{-1}(L_0,\hat{c})L_{-1}\!\ket{h}=L_{-1}\frac{1}{2 L_0}L_1 L_{-1}\!\ket{h}\\=L_{-1}\frac{2L_0}{2L_0}\!\ket{h}=L_{-1}\!\ket{h},
\end{multline}
which is true for any Verma module.

Let us now take as induction hypothesis that expression~\eqref{eq:inverse_Shapovalov_form} acts appropriately as the identity for any level $<\check{\ell}$. We will now show that this implies it must act as the identity also on any level-$\check\ell$ descendant.

When considering the action on any descendant $L_{-\rho}\!\ket{h}$ with $\abs{\rho}=\check\ell$ we can distinguish two cases, depending on whether the state we are acting on is a (descendant of a) singular vector or not.

\textbf{Case I.}\,\,\, If the state we act on is a singular vector $L_{\expval{r,s}}\ket{h_{\expval{r,s}}}$ of level $r s\le \check\ell$ or a descendant thereof, we can take advantage of Lemma~\ref{lemma2}. This instructs us that, for any surviving term, the reordering procedure has to produce a factor $(L_0-h_{\expval{r,s}})^n$, which on its own would be vanishing when evaluated on $\ket{h_{\expval{r,s}}}$. The only terms that can then survive, are those that contain a pole that precisely cancels this zero. By Lemma~\ref{lemmaBrown}, we know that there can only be simple poles in the inverse Shapovalov form, so the power $n$ has to be equal to one. We thus can focus only on terms with a simple pole $(L_0-h_{\expval{r,s}})^{-1}$. This divides further into two subcases.

\textbf{Subcase a.} \,\,\, If the state we are acting upon is itself a singular vector, \textit{i.e.}\ $rs=\check\ell$, there is only one term of level $\check{\ell}$ in~\eqref{eq:inverse_Shapovalov_form} with the appropriate pole. Using~\eqref{eq:singular_vector_decomposition}, its action on $L_{\expval{r,s}}\ket{h_{\expval{r,s}}}$ corresponds to
\begin{equation}
\begin{aligned}
    L_{\expval{r,s}}& \frac{q_{\expval{r,s}}}{L_0-h_{\expval{r,s}}}L_{\expval{r,s}}^\dagger L_{\expval{r,s}} \ket{h_{\expval{r,s}}}\\
    &=L_{\expval{r,s}}\!\frac{q_{\expval{r,s}}v_{\expval{r,s}}^\mu v_{\expval{r,s}}^\nu}{L_0-h_{\expval{r,s}}}\! \left[S_{r s}(L_0,\hat{c})\right]_{\mu\nu}\!\ket{h_{\expval{r,s}}}.
    \end{aligned}
    \label{eq:action_identity_one_singvec}
\end{equation}
It is easy to check that this action reduces to that of the identity once we replace $q_{\expval{r,s}}$ with its definition~\eqref{eq:qrs_as_inverse_norm}. While this proves sufficiency, one may wonder if the level-$\check{\ell}$ inverse Shapovalov form can contain additional terms with the same pole $\left(L_0-h_{\expval{r,s}}\right)^{-1}$. The answer is negative, as we will now argue. In fact, to be able to reconstruct the $L_{\expval{r,s}}\ket{h_{\expval{r,s}}}$ singular vector after deconstructing it with positive modes, we necessarily need to have $L_{\expval{r,s}}$ as the only creation operator. The Hermiticity of the inverse Shapovalov form then forces us to include only $L_{\expval{r,s}}^{\dagger}$ on its right, reproducing precisely the term we analyzed in~\eqref{eq:action_identity_one_singvec}.
We have therefore proven that for every singular vector of level $rs=\check{\ell}$ there must be exactly one term with the corresponding pole in the inverse Shapovalov form, and that said term corresponds precisely to what is described by~\eqref{eq:inverse_Shapovalov_form} and~\eqref{eq:qrs_as_inverse_norm}.

\textbf{Subcase b.} \,\,\,If we act on a descendant of a singular vector $L_{-\rho} L_{\expval{r,s}} \!\ket{h_{\expval{r,s}}}$ with $rs<\check\ell$ and $\abs{\rho}+r s=\check{\ell}$, the action of our candidate identity can be written as
\begin{equation}
\begin{aligned}
    \sum_{\tilde{\boldsymbol{r}}\cdot\tilde{\boldsymbol{s}}=\check{\ell}-r s}&L_{\expval{r_{m},s_{m}}}\dots L_{\expval{r_{2},s_{2}}}\\
    &\times \left(L_{\expval{r,s}}\frac{q_{\expval{r,s},\expval{r_2,s_2},\dots \expval{r_{m},s_{m}}}}{L_0-h_{\expval{r,s}}}L_{\expval{r,s}}^{\dagger}\right)\\
    &\times L_{\expval{r_{2},s_{2}}}^{\dagger}\dots L_{\expval{r_{m},s_{m}}}^{\dagger}L_{-\rho} L_{\expval{r,s}} \!\ket{h_{\expval{r,s}}},
    \end{aligned}
\end{equation}
where $\tilde{\boldsymbol{r}}=(r_2,\dots,r_m)$ and $\tilde{\boldsymbol{s}}=(s_2,\dots,s_m)$ parametrize~$\compls_{\check{\ell}-r s}$.
We can now observe that the product $L_{\expval{r_{2},s_{2}}}^{\dagger}\dots L_{\expval{r_{m},s_{m}}}^{\dagger}L_{-\rho}$, which is of total level zero and is acting on the singular vector $L_{\expval{r,s}} \!\ket{h_{\expval{r,s}}}$, must actually equate to a polynomial of $L_0$ and $\hat{c}$. This means then that it must commute with the term in parentheses, which is also of total level zero. Performing this commutation and multiplying and dividing by the $q_{\expval{r,s}}$ coefficient, we get
\begin{multline}
    \sum_{\tilde{\boldsymbol{r}}\cdot\tilde{\boldsymbol{s}}=\check{\ell}-r s}\frac{q_{\expval{r,s},\expval{r_2,s_2},\dots \expval{r_{m},s_{m}}}}{q_{\expval{r,s}}}\\
    \times L_{\expval{r_{m},s_{m}}}\dots L_{\expval{r_{2},s_{2}}}L_{\expval{r_{2},s_{2}}}^{\dagger}\dots L_{\expval{r_{m},s_{m}}}^{\dagger}L_{-\rho}\\
    \times \!\left(\!L_{\expval{r,s}}\frac{q_{\expval{r,s}}}{L_0-h_{\expval{r,s}}}L_{\expval{r,s}}^{\dagger}\!\right)\!L_{\expval{r,s}} \!\ket{h_{\expval{r,s}}\!}.
\end{multline}
It is now easy to identify the third line with the action of the identity we just analyzed in~\eqref{eq:action_identity_one_singvec}. By inductive assumption, the term in parentheses acts then trivially and can thus be erased, leaving us with
\begin{multline}
    \sum_{\tilde{\boldsymbol{r}}\cdot\tilde{\boldsymbol{s}}=\check{\ell}-r s}\frac{q_{\expval{r,s},\expval{r_2,s_2},\dots \expval{r_{m},s_{m}}}}{q_{\expval{r,s}}} L_{\expval{r_{m},s_{m}}}\dots L_{\expval{r_{2},s_{2}}}\\
    \times L_{\expval{r_{2},s_{2}}}^{\dagger}\dots L_{\expval{r_{m},s_{m}}}^{\dagger}L_{-\rho}L_{\expval{r,s}} \!\ket{h_{\expval{r,s}}}.
\end{multline}
Using~\eqref{eq:recurrence_q_coefficients} we identify the ratio of $q$-coefficients with
\begin{equation}
    \frac{q_{\expval{r,s},\expval{r_2,s_2},\dots \expval{r_{m},s_{m}}}}{q_{\expval{r,s}}}=\frac{q_{\expval{r_2,s_2},\dots \expval{r_{m},s_{m}}}}{h_{\expval{r,s}}+r s-h_{\expval{r_2,s_2}}}\,,
\end{equation}
and we furthermore have the identity
\begin{multline}
    \frac{1}{h_{\expval{r,s}}+r s-h_{\expval{r_2,s_2}}} L_{\expval{r,s}} \!\ket{h_{\expval{r,s}}}\\
    =\frac{1}{L_0-h_{\expval{r_2,s_2}}}L_{\expval{r,s}} \!\ket{h_{\expval{r,s}}}.
\end{multline}
Moving the latter fraction through the product $L_{\expval{r_{2},s_{2}}}^{\dagger}\dots L_{\expval{r_{m},s_{m}}}^{\dagger}L_{-\rho}$ (again, of total level zero), we reconstruct the inverse Shapovalov at level $\abs{\rho}<\check\ell$
\begin{equation}
\begin{split}
    \begin{aligned}
    \Biggl(\,\sum_{\tilde{\boldsymbol{r}}\cdot\tilde{\boldsymbol{s}}=\check{\ell}-r s}& L_{\expval{r_{m},s_{m}}}\dots L_{\expval{r_{2},s_{2}}} \frac{q_{\expval{r_2,s_2},\dots \expval{r_{m},s_{m}}}}{L_0-h_{\expval{r_2,s_2}}}\\
    &\times L_{\expval{r_{2},s_{2}}}^{\dagger}\dots L_{\expval{r_{m},s_{m}}}^{\dagger}\Biggr)L_{-\rho} L_{\expval{r,s}} \!\ket{h_{\expval{r,s}}}\end{aligned}\\
    = \mathbf{S}^{-1}_{\abs{\rho}}(L_0,\hat{c}) L_{-\rho} L_{\expval{r,s}} \!\ket{h_{\expval{r,s}}}\\
    =L_{-\rho} L_{\expval{r,s}} \!\ket{h_{\expval{r,s}}},
\end{split}
\end{equation}
where we used~\eqref{eq:inverse_Shapovalov_form} in the first equality and~\eqref{eq:inverseShap_identity_on_descendants} with $\ket{\psi}\equiv L_{\expval{r,s}}\!\ket{h_{\expval{r,s}}}$ in the second. We have thus shown that also in this case the RHS of~\eqref{eq:inverse_Shapovalov_form} acts as the identity and no other terms with the same pole can be present. Since the set of poles of the inverse Shapovalov form should manifest in the set of zeros of the Kac determinant~\eqref{eq:Kacdeterminant}, and since Lemma~\ref{lemmaBrown} guarantees us that there can only be simple poles, what we have shown so far also implies that no other singularities can be present in~\eqref{eq:inverse_Shapovalov_form}.

\textbf{Case II.} \,\,\, If the state we are acting on is neither singular nor a descendant of a singular vector, no particular set of terms gets isolated, and all the terms of~\eqref{eq:inverse_Shapovalov_form} must contribute.

As we argued above, there cannot be additional singular terms, so the only possible way in which~\eqref{eq:inverse_Shapovalov_form} could have corrections is if there are contributions from some regular terms. Regular terms in~\eqref{eq:inverse_Shapovalov_form} would acquire the form
\begin{equation}
    L_{-\mu}[f_\ell(L_0,\hat{c})]^{\mu\nu}L_\nu,
\end{equation}
with $|\mu|=|\nu|=\ell$ and $[f_\ell(L_0,\hat{c})]^{\mu\nu}$ regular functions of $L_0$.  Acting on the descendant $L_{-\ell}\!\ket{h}$ for a generic $h$ leads to
\begin{equation}
\begin{aligned}
&L_{-\mu}[f_\ell(L_0,\hat{c})]^{\mu\nu}L_\nu L_{-\ell}\!\ket{h}\\
&=L_{-\mu}\ket{h}[f_{\ell}(h,c)]^{\mu\nu}[S_{\ell}(h,c)]_{(\ell)\nu}\\
&=L_{-\mu}\ket{h}[f_{\ell}(h,c)]^{\mu\nu}\left(a_{(\ell)\nu} h+ \frac{c}{12}b_{(\ell)\nu}\right),
\end{aligned}
\end{equation}
where the $a_{(\ell)\nu}\ne 0$ and $b_{(\ell)\nu}$ are just integers.  Clearly, the terms above diverge at least linearly for $h\to\infty$. This divergence cannot be canceled by the singular terms we have already studied since, when acting on $L_{-\ell}\ket{h}$, these can only produce terms that scale at most as~$h^0$. This means that all regular terms must vanish.

Having considered all possible terms with poles and having excluded the presence of regular terms, we can conclude that the expression~\eqref{eq:inverse_Shapovalov_form} for the index-free inverse Shapovalov form is correct and complete.

\section{Concluding Remarks}

In this paper, we introduced and proved a new expression for the index-free inverse Shapovalov form of the Virasoro algebra for generic central charge, as displayed in~\eqref{eq:inverse_Shapovalov_form}. Although it may seem unconventional at first, this index-free version of the inverse Shapovalov form is what appears explicitly in the Virasoro completeness relation, \textit{cf.}~\eqref{eq:resol_identity}, and has various desirable properties.
To begin with, the pole structure of the inverse Shapovalov is manifest, with only presence of simple poles in the locations dictated by the Kac determinant. The residues at those poles are all made of (products of) singular vectors and their conjugates, which are the only way creation and annihilation modes are packaged in the inverse Shapovalov form.

Furthermore, the form of expression~\eqref{eq:inverse_Shapovalov_form} is particularly interesting when viewed through the lens of the sewing procedure~\cite{Sonoda:1988mf}. Provided knowledge of singular vectors, our new formula reduces the level-by-level computation of generic conformal blocks to the action of the differential operators associated with the singular vectors $L_{\expval{r,s}}$ on three-point correlators.  As an example, for four-point conformal blocks this procedure yields
\begin{equation*}
\begin{aligned}
    F_h(\eta)&=\eta^h\sum_{\ell\geq0}\sum_{\boldsymbol{r}\cdot\boldsymbol{s}=\ell}\frac{q_{\expval{r_1,s_1},\dots, \expval{r_{m},s_{m}}}}{h-h_{\expval{r_1,s_1}}} \eta^{\ell}\\
    &\times\prod_{j=1}^m\left[h-h_1+\sum_{k=1}^{j-1}r_ks_k,\,h_2\right]_{\expval{r_j,s_j}}\\
    &\times\prod_{j=1}^m\left[h-h_4+\sum_{k=1}^{j-1}r_ks_k,\,h_3\right]_{\expval{r_j,s_j}}
\end{aligned},
\end{equation*}
in the notation of~\cite{Osborn:2012vt}, and where
\begin{equation*}
    [a,b]_\mu=\prod_{j=1}^n\left[a+\sum_{k=1}^{j-1}\mu_k+b\mu_j\right]
\end{equation*}
is a generalization of the usual Pochhammer symbols for two numbers $a$, $b$ as well as a partition $\mu=(\mu_1,\dots,\mu_n)$ and $[a,b]_{\expval{r,s}}=v_{\expval{r,s}}^\mu[a,b]_\mu$.  More details will be provided in an upcoming work~\cite{In_preparation}.

Being structured as a sum of diagonal terms, expression~\eqref{eq:inverse_Shapovalov_form} also reorganizes the sewing procedure in a way compatible with the operator product expansion (OPE). Applying in fact~\eqref{eq:correlator_decomp_shapovalov} recursively to decompose a correlator into sums of products of three-point correlators, one obtains a form reminiscent of the OPE decomposition of correlators modulo the OPE coefficients. This suggests that the coefficients that appear in~\eqref{eq:inverse_Shapovalov_form} could be related to the kinematic coefficients that dictate the form of the differential operator in the OPE.

Despite the fact that the index-free inverse  Shapovalov~\eqref{eq:inverse_Shapovalov_form} can be used directly in the discussed cases, one may be interested in its matrix form. In this case, the expression we introduced can still be of great use, as it allows to extract the matrix elements by a simple reordering of the products of singular vectors, rather than computing the inverse of large matrices.

To conclude with, it is clear that our expression can lead to new interesting insights on some aspects of the Virasoro algebra. As another example, our result can also be used to determine the form of a generalized Casimir operator for the Virasoro algebra, as we will address in an upcoming publication~\cite{Fortin:2024wcs}.

\begin{acknowledgments}
JFF would like to thank Pierre Mathieu for enlightening discussions and CERN for its hospitality.  This work was supported by NSERC (JFF and LQ) and the US Department of Energy under grant DE-SC00-17660 (WS).
\end{acknowledgments}

\appendix

\section{The Regularized Norm of Singular Vectors and Zamolodchikov's Expression}
\label{app:Zamolodchikov_proof}

In~\cite{Zamolodchikov:2003yb}, Zamolodchikov proposed a formula for the regularized norm of singular vectors, whose inverse appears in the RHS of~\eqref{eq:qrs_as_inverse_norm}. This formula, also proven algebraically in~\cite{Yanagida:2010qm}, reads as follows in our conventions:
\begin{equation}
    q_{\expval{r,s}}^{-1}=2 \prod_{\substack{(j,k)\in \mathbbb{Z}^2\\1-r\le j\le r\\
    1-s\le k \le s\\
    (j,k)\ne (0,0),(r,s)}}\left(j t^{\frac12}+k t^{-\frac12}\right).
    \label{eq:Zamolodchikov_norm}
\end{equation}
We will now show that this is equivalent to the inverse of~\eqref{eq:single_vector_coefficients}. Naming $a_{j,k}$ the expression in brackets above, we can decompose the product range as 
\begin{equation}
    q_{\expval{r,s}}^{-1}=2\hspace*{-10pt}\prod_{\substack{1 \le j \le r\\
    1\le k\le s\\
    (j,k)\ne (r,s)}}\hspace*{-6pt}a_{j,k} \prod_{\substack{1-r \le j \le 0\\
    1\le k\le s}}\hspace*{-4pt}a_{j,k} \prod_{\substack{1 \le j \le r\\
    1-s\le k\le 0}}\hspace*{-4pt}a_{j,k} \prod_{\substack{1-r \le j \le 0\\
    1-s\le k\le 0\\
    (j,k)\ne (0,0)}}\hspace*{-6pt}a_{j,k}\,.
\end{equation}
Naming the four products above respectively $A_1$ to $A_4$, the following holds:
\begin{align}
    &A_1=\frac{t^{\frac{r s+1}{2}}(1+rt)_{s-1}}{s!}\prod_{k=1}^s\left(k/t\right)_r\,,\\
    &A_2=(-1)^{rs} t^{\frac{r s}{2}}\prod_{k=1}^s\left(-k/t\right)_r\,,
    \end{align}
    \begin{align}
    &A_3=\left[(-1)^{rs} t^{\frac{r s}{2}}\right]^{-1}\prod_{j=1}^r\left(-j t\right)_s\,,\\
    &A_4=\left[\frac{t^{\frac{r s+1}{2}}(1+rt)_{s-1}}{s!}\right]^{-1}\frac{(-1)^{rs+1}}{rs}\prod_{j=1}^r\left(j t\right)_s\,.
\end{align}
At this point, simply taking the product of these four objects times two reproduces the inverse of~\eqref{eq:single_vector_coefficients}.


\end{document}